\documentclass[letterpaper, 10 pt, conference]{ieeeconf}
\IEEEoverridecommandlockouts
\usepackage[usenames,dvipsnames,svgnames,table]{xcolor}
\usepackage{amssymb,latexsym,amsfonts,amsmath}
\usepackage{graphicx}
\usepackage{eso-pic}
\usepackage{epsfig}
\usepackage{mathtools}
\usepackage{tikz}
\usetikzlibrary{automata}
\usetikzlibrary{shapes}
\usetikzlibrary{calc,shapes,arrows}
\usepackage{mdwlist}
\usepackage{refcount}
\usepackage{comment}
\usepackage{ragged2e}
\usepackage[noadjust]{cite}

\let\labelindent\IEEElabelindent
\usepackage{enumitem}

\makeatletter
\def\endthebibliography{%
	\def\@noitemerr{\@latex@warning{Empty `thebibliography' environment}}%
	\endlist
}
\makeatother

\usepackage{keyval,trig}
\usepackage{amssymb,dsfont}
\usepackage{mathtools,mathrsfs}

\definecolor{mediumblue}{rgb}{0.0, 0.0, 0.8}
\definecolor{mediumcandyapplered}{rgb}{0.89, 0.02, 0.17}
\definecolor{nazar}{rgb}{0.7, 0.5, 0.9}
\makeatletter
\let\NAT@parse\undefined
\makeatother
\usepackage[colorlinks=true, citecolor=mediumblue, linkcolor=mediumblue, urlcolor = nazar, final]{hyperref}

\definecolor{lightblue}{rgb}{0.30,0.75,0.93}
\newcommand{\greensquare}{\tikz\fill[green!20!white] (0,0) rectangle (2mm,2mm);}
\newcommand{\redsquare}{\tikz\fill[red!20!white] (0,0) rectangle (2mm,2mm);}
\newcommand{\reddsquare}{\tikz\fill[red!70!white] (0,0) rectangle (2mm,2mm);}
\newcommand{\bluesquare}{\tikz\fill[lightblue!60!white] (0,0) rectangle (2mm,2mm);}

\newcommand{\dashedline}[1]{\tikz \draw [black, thick, dashed] (0,0) -- (#1,0);}
\newcommand{\dashedlinea}[1]{\tikz \draw [blue, thick, dashed] (0,0) -- (#1,0);}

\DeclareRobustCommand\sampleline[1]{%
	\tikz\draw[#1] (0,0) (0,\the\dimexpr\fontdimen22\textfont2\relax)
	-- (2em,\the\dimexpr\fontdimen22\textfont2\relax);%
}

\usepackage{algorithm}
\usepackage{algorithmic}

\usepackage{tikz}
\usetikzlibrary{calc,shapes,arrows}

\usepackage{xspace}

\newcounter{theorem}
\newcounter{definition}
\newcounter{lemma}
\newcounter{claim}
\newcounter{problem}
\newcounter{proposition}
\newcounter{corollary}
\newcounter{construction}
\newcounter{example}
\newcounter{xca}
\newcounter{comments}
\newcounter{remark}
\newcounter{assumption}

\newtheorem{theorem}[theorem]{Theorem}
\newtheorem{lemma}[lemma]{Lemma}

\newtheorem{problem}[problem]{Problem}

\newtheorem{definition}[definition]{Definition}

\newtheorem{remark}[remark]{Remark}

\usepackage[many]{tcolorbox}
\usetikzlibrary{calc}
\tcbuselibrary{skins}

\newtcolorbox{resp}[1][]{%
	enhanced jigsaw,%
	colback=gray!5!white,%
	colframe=gray!80!black,%
	size=small,%
	boxrule=1pt,%
	halign title=flush center,%
	coltitle=black,%
	breakable,%
	drop shadow=black!50!white,%
	attach boxed title to top left={xshift=1cm,yshift=-\tcboxedtitleheight/2,yshifttext=-\tcboxedtitleheight/2},%
	minipage boxed title=3cm,%
	boxed title style={%
		colback=white,%
		size=fbox,%
		boxrule=1pt,%
		boxsep=2pt,%
		underlay={%
			\coordinate (dotA) at ($(interior.west) + (-0.5pt,0)$);
			\coordinate (dotB) at ($(interior.east) + (0.5pt,0)$);
			\begin{scope}[gray!80!black]
				\fill (dotA) circle (2pt);
				\fill (dotB) circle (2pt);
			\end{scope}
		}%
	},%
	#1%
}

\xdefinecolor{MATLABblue}{rgb}{0.0,0.45,.74}
\xdefinecolor{MATLABred}{rgb}{0.85,0.33,.1}

\newcommand{\R}{{\mathbb{R}}}

\newcommand{\N}{{\mathbb{N}}}

\IEEEoverridecommandlockouts

\usepackage[normalem]{ulem}

\makeatletter
\def\@opargbegintheorem#1#2#3{\textit{#1\ #2} \textit{(#3):}}
\makeatother

\newcommand{\CLbr}{\tilde{A} \tilde{\mathbb{R}}(x)\mathbb{Q}(x)}
\title{Learning Robust Safety Controllers for Uncertain Input-Affine Polynomial Systems
\thanks{}}

\author{Omid Akbarzadeh, \IEEEmembership{Student Member,~IEEE}, MohammadHossein Ashoori,  \IEEEmembership{Student Member,~IEEE}, \\and Abolfazl Lavaei, \IEEEmembership{Senior Member,~IEEE}
\thanks{All the authors are with the School of Computing, Newcastle University, United Kingdom. Emails: \{\texttt{\href{mailto:o.akbarzadeh2@newcastle.ac.uk}{o.akbarzadeh2},\href{mailto:m.ashoori2@newcastle.ac.uk}{m.ashoori2}, \href{mailto:abolfazl.lavaei@newcastle.ac.uk}{abolfazl.lavaei}}\}\texttt{@newcastle.ac.uk}}
}
\usepackage{etoolbox}
\makeatletter
\pretocmd{\caption}{\gdef\@currentlabel{\arabic{figure}}}{}{}
\makeatother
\begin{document}
\maketitle

\begin{abstract}%
This paper offers a \emph{direct data-driven} approach for learning \emph{robust control barrier certificates (R-CBCs)} and \emph{robust safety controllers (R-SCs)} for discrete-time input-affine polynomial systems with \emph{unknown dynamics} under \emph{unknown-but-bounded disturbances}. The proposed method relies on data from input-state observations collected over a finite-time horizon while satisfying a specific \emph{rank condition} to ensure the system is persistently excited. Our data-driven scheme enables the synthesis of R-CBCs and R-SCs directly from observed data, bypassing the need for explicit modeling of the system's dynamics and thus ensuring robust system safety against disturbances within an infinite time horizon. Our proposed approach is formulated as a \emph{sum-of-squares (SOS)} optimization problem, providing a structured design framework. Two case studies showcase our method's capability to provide robust safety guarantees for unknown input-affine polynomial systems under bounded disturbances, demonstrating its practical effectiveness.
\end{abstract}

\section{Introduction}
Many complex and heterogeneous systems, particularly those in safety-critical applications~\cite{mcgregor2017analysis}, lack explicit mathematical models due to their inaccessibility or excessive complexity. This absence of models poses significant challenges for traditional model-based methods, which rely on detailed system knowledge for effective analysis and control of dynamical systems. To address this critical challenge, the past decade has witnessed a remarkable growth in data-driven approaches for analyzing and controlling systems with unknown dynamics. These methods utilize observed data to gain insights into system behavior, enabling effective analysis without the need for explicit mathematical models.

Data-driven techniques are generally classified into two categories: \emph{indirect} and \emph{direct} methods~\cite{Hou2013model, de2019formulas,dorfler2022bridging}. Indirect methods begin by approximating the unknown system through system identification techniques, aiming to construct a mathematical model from observed data. This model facilitates the use of traditional model-based control strategies, leveraging established theories and tools in the existing model-based literature. However, the \emph{two-step} process, which involves accurately identifying the system and subsequently applying control methods to the derived model, can be computationally demanding, especially for real-life applications with complex dynamics~(\emph{e.g.,}~\cite{Indrict-data-driven1,Indrict-data-driven}).

In contrast, \emph{direct} data-driven methods circumvent the need for explicit model construction by utilizing system measurements directly to perform analysis and control tasks, effectively operating on the data without translating it into a traditional model. By eliminating the model identification phase, direct methods reduce computational overhead and minimize potential inaccuracies caused by model approximation errors. This makes them especially beneficial for complex systems where constructing an explicit model is infeasible or computational resources are constrained~\cite{TANASKOVIC20171,direct-data}.

In recent years, ensuring the safety of complex dynamical systems has emerged as a key focus in control theory. To address this problem, various methods have been introduced in the literature, including the concept of control barrier certificates~\cite{10.1007/978-3-540-24743-2_32,wieland2007constructive}. In particular, barrier certificates operate by imposing a set of inequalities on the system’s state and dynamics, analogous to a Lyapunov function. By defining a level set that effectively prevents system trajectories from entering unsafe regions, starting from a specified initial set, a barrier certificate provides a formal safety guarantee over the dynamical system. However, a key limitation of these methods is their reliance on having a precise mathematical model of the system \cite{ames2019control,clark2021control,Mahathi,nejati2024reactive,lavaei2024scalable}. 

To address this challenge, recent developments have centred on data-driven techniques that offer safety assurances for systems with unknown dynamics using \emph{scenario approaches}~\cite{calafiore2006scenario,salamati2024data,nejati2023data}. However, these methods rely on data being independent and identically distributed (i.i.d.), necessitating the collection of numerous trajectories—potentially millions in practical scenarios—from the unknown system. In contrast, our approach adopts a \emph{non-i.i.d.} trajectory-based framework, requiring only a \emph{single input-state trajectory} from the unknown system over a specific time horizon, provided that the system is persistently excited~\cite{willems2005note,samari2025data} (see \cite{markovsky2021behavioral} for a recent review on the fundamental lemma related to persistency of excitation). This method offers a significant advantage by eliminating the need for multiple independent trajectories, making it more practical in situations where obtaining several distinct trajectories is challenging.

Some recent studies, including~\cite{LUPPI2024100914,9966144, BISOFFI20203953}, have proposed data-driven methods to design state-feedback controllers that ensure robust invariance of compact polyhedral sets containing the origin. Although these approaches show promise, they tend to be conservative, as controllers might be feasible for \emph{certain regions} within the polyhedral set, even if a solution for the whole set is unattainable. Furthermore, while the methods in~\cite{LUPPI2024100914,9966144, BISOFFI20203953} focus on achieving robust invariance of a compact polyhedral set, our approach adopts a more general notion of safety. Specifically, our method accommodates \emph{multiple} unsafe regions and ensures that the unknown system, starting from an initial set, avoids reaching these unsafe regions (cf. both case studies).

It is worth noting that most data-driven techniques in the literature that use trajectory-based approaches, (\emph{e.g.,}~\cite{LUPPI2024100914,nejati2022data,zaker2025data,akbarzadeh2024data}), have been developed for  control analysis of \emph{continuous-time} nonlinear systems. While it may seem that addressing continuous-time systems poses more challenges, developing data-driven methods via the trajectory-based approach for \textit{discrete-time} unknown systems proves to be even more complex (see~\cite{martin2023guarantees}). Specifically, the approach proposed in the literature (\emph{e.g.,} \cite{nejati2022data} in the continuous-time setting) suggests $\mathcal{B}(x) = \mathcal{R}(x)^\top P \mathcal{R}(x)$ as a candidate control barrier certificate, where $\mathcal{R}(x)$ represents the system's monomials. The corresponding one-step ahead of the barrier function is $\mathcal{B}(x^+) = \mathcal{R}(x^+)^\top P \mathcal{R}(x^+)$, where $x^+$ is the one-step evolution of the system in discrete time. However, this approach encounters a fundamental obstacle for \emph{discrete-time} nonlinear polynomial systems, as the closed-loop data-based representation of $\mathcal{R}(x^+)$ is not straightforward. 

While a recent work \cite{samari2024} explores a data-driven approach for discrete-time nonlinear systems, our method differs in three key aspects. First, \cite{samari2024} focuses solely on discrete-time nonlinear polynomial systems, whereas our approach provides \emph{robust safety guarantees} for systems subject to \emph{unknown-but-bounded disturbances}. Second, \cite{samari2024} addresses only input-affine systems with a static control matrix $B$, whereas our method allows this matrix to be state-dependent, denoted by $\mathcal G(x)$, making it more adaptable to a broader range of scenarios. Lastly, the approach in \cite{samari2024} employs a conservative notion of barrier certificates, requiring the barrier to decay at each time step. In contrast, inspired by the $c$-martingale property in stochastic systems, we adopt a more relaxed barrier that permits decay relative to a positive constant $c$ (cf. condition \eqref{subeq: decreasing}), thus significantly increasing the likelihood of finding such a function in practice.

{\bf Primary Contributions.} Motivated by the discussed challenges, we develop a \emph{direct} data-driven methodology for discrete-time input-affine polynomial systems with \emph{unknown dynamics} and \emph{unknown-but-bounded} disturbances, enabling the design of R-CBCs and corresponding R-SCs directly from observed input-state data.  Our method relies on data collected over a finite time horizon, while fulfilling a specific \emph{rank condition} to ensure the system remains persistently excited~\cite {willems2005note}. By leveraging this data, our approach directly constructs R-CBCs and R-SCs, thereby providing robust safety guarantees against disturbances within an infinite time horizon. This is achieved through a structured framework based on SOS optimization. We apply our approach to two case studies, showcasing our method's ability to ensure robust safety for unknown input-affine polynomial systems under \emph{bounded disturbances}.

\section{System Description}
\subsection{Notation}
The sets of real numbers, non-negative and positive real numbers are denoted by $\mathbb{R}$, $\mathbb{R}^+_0$, and $\mathbb{R}^+$\!, respectively. Similarly, $\mathbb{N} = \{0, 1, 2, \ldots\}$ and $\mathbb{N}^+ = \{1, 2, \ldots\}$ represent the sets of non-negative and positive integers, respectively. The identity matrix of size $n \times n$ is expressed as $\mathds{I}_n$. Given \( N \) vectors \( x_i \in \mathbb{R}^{n_i} \), the column vector \( x \) is formed by stacking the vectors as \( x = [x_1; \dots; x_N] \). Moreover, we use $[x_1\;\; \ldots\;\; x_N]$ to represent horizontal concatenation of vectors $x_i\in\R^n$ to form an $n\times N$ matrix. A \emph{symmetric} matrix $P$ is positive definite when $P \succ 0$, while $P \succeq 0$ signifies that $P$ is a \emph{symmetric} positive semi-definite matrix. The transpose of a matrix $P$ is written as $P^\top$, and its left pseudoinverse is denoted by $P^\dagger$. In symmetric matrices, the symbol $(\star)$ represents the transposed element in the corresponding symmetric position. We use $\Vert \cdot \Vert$ to denote the Euclidean norm for a vector $x \in \mathbb{R}^n$ and the induced 2-norm of a matrix $A \in \mathbb{R}^{n \times m}$. The empty set is denoted by $\emptyset$. 
For a system $\Theta$ and a property $\Upsilon$, the notation $\Theta \models \Upsilon$ signifies that $\Theta$ satisfies $\Upsilon$.
\subsection{Discrete-Time Input-Affine Polynomial Systems}
First, we introduce the definition of discrete-time input-affine polynomial systems as follows.

\begin{definition}[\textbf{dt-IAPS}]\label{def: dt-IAPS}
A discrete-time input-affine polynomial system (dt-IAPS) is described as
\begin{equation}\label{sys}
\Theta\!: x^+= A\mathcal{R}(x) + B\mathcal{G}(x)u + w, 
\end{equation}
where $x^+$ is the state variable at the next time step, \emph{i.e.,} $x^+ := x(k + 1),$ with $k \in \N$. Moreover, $x \in X$ with $X \subseteq \mathbb{R}^n $ representing the state set, $ u \in U$ with $U \subseteq \mathbb{R}^m $ denoting the input set, while \( \mathcal{R}(x) \in \mathbb{R}^{N} \) and \(\mathcal{G}(x) \in \mathbb{R}^{\hat{N} \times m}\) are a vector and a matrix of monomials in states $x$, respectively. Additionally, \( A \in \mathbb{R}^{n \times N} \) and \( B \in \mathbb{R}^{n \times \hat{N}} \) are the system and control matrices, while \( w \in \mathbb{R}^n \) represents the system disturbance. We denote by $x_{x_{0uw}}(k)$ the value of the solution process of $\Theta$ at time $k \in \mathbb{N}$ , under input and disturbance trajectories $u(\cdot)$ and $w(\cdot)$, starting from an initial state $x_0\in X$.
\end{definition}

We assume that the matrices $A$ and $B$ are \emph{unknown}, as is often the case in real-world scenarios. 
\color{black}
Additionally, the disturbance $w$ belongs to the following set
\begin{equation}\label{assum}
 W(\delta)=\left\{w\in \mathbb{R}^n \mid w^{\top} w \leq \delta, ~\delta \in \mathbb{R}^{+}_0 \right\}\!,  
\end{equation}  
which is commonly referred to as \emph{unknown-but-bounded} disturbances with an \emph{instantaneous bound} on the norm of $w$~\cite{bisoffi2021trade}. While $\mathcal{R}(x)$ and $\mathcal{G}(x)$ are unknown, we assume that a dictionary for them can be constructed to capture the actual dynamics by being \emph{sufficiently extensive}. Specifically, such dictionaries are structured to be comprehensive, capturing all possible terms in the true system's dynamics, even if they include some terms that are ultimately unnecessary. To achieve this, an upper bound on the maximum degree of $\mathcal{R}(x), \mathcal{G}(x)$ can be obtained based on the physical insights of the system, allowing $\mathcal{R}(x), \mathcal{G}(x)$ to be constructed to include all possible combinations of states up to that bound (cf. both case studies).

We now define the robust safety problem for the dt-IAPS $\Theta$ described in~\eqref{sys}.
\begin{definition}\label{safey}
 Consider a safety specification $\Upsilon = (X_0, X_1)$, where $X_0, X_1 \subseteq X$ are initial and unsafe regions with $X_0 \cap X_1 = \emptyset$. A dt-IAPS is said to be robustly safe against the unknown-but-bounded disturbance $w$, expressed as $\Theta \models \Upsilon$, if all trajectories of $\Theta$ originating from the initial set $X_0$ never enter the unsafe set $X_1$ within an infinite time horizon.   
\end{definition}
Since the main objective of this work is to design a robust safety certificate for the unknown dt-IAPS~\eqref{sys}, we formally define \textit{robust control barrier certificates} in the next subsection.

\subsection{Robust Control Barrier Certificates}
\begin{definition}[\textbf{R-CBC}]\label{def: R-CBC}
	Consider a dt-IAPS
	$\Theta$, with $X_0, X_1 \subseteq X$ being its initial and unsafe sets, respectively. Assuming the existence of constants $\gamma_1,\gamma_2 \in \mathbb{R}^{+}$, $\rho\in \mathbb{R}^{+}_0$, and $\lambda \in (0,1]$ with $\gamma_2 > \gamma_1$, a function $\mathcal B:X \to \mathbb{R}_0^+$ is called a robust control barrier certificate (R-CBC) for $\Theta$ if
	
	\begin{subequations}\label{eq: R-CBC}
		\begin{align}
			&  \:\:  \mathcal B(x) \leq \gamma_1, \hspace{2.55cm} \forall x \in X_{0},\label{subeq: initial}\\
			&  \:\:  \mathcal B(x) \geq \gamma_2, \hspace{2.55cm} \forall x \in X_{1}, \label{subeq: unsafe}
		\end{align}  
		and $\forall x\in \tilde{X} = \big \{ x\in X\!\!: \mathcal{B}(x) < \gamma_2\big \}$, $\exists u\in U$, such that $\forall w \in W(\delta)$ 
		\begin{align}\label{subeq: decreasing}
			\mathcal B(x^+) \leq  \lambda\mathcal B(x) + \rho \Vert w\Vert^2,
		\end{align}
	\end{subequations}
    with $c = \rho \delta$ satisfying \vspace{-0.2cm}
\begin{align}\label{eq: climit}
	c \leq  \gamma_2 (1-\lambda).
\end{align}
	Accordingly, $u$ satisfying \eqref{subeq: decreasing} is a robust safety controller (R-SC) for the dt-IAPS for infinite time horizon.
\end{definition}

\begin{remark}
Drawing inspiration from the $c$-martingale properties in stochastic systems~\cite{1967stochastic}, we introduce a relaxation in \eqref{subeq: decreasing} by allowing the R-CBC to decay relative to a constant $c = \rho \delta$ (note that $\rho \Vert w\Vert^2 \leq  \rho \delta = c$ according to \eqref{assum}). This adjustment significantly improves the feasibility of finding an R-CBC while still ensuring robust safety over an infinite time horizon.
\end{remark}
\color{black}
\begin{remark}
Our definition of R-CBC accommodates multiple unsafe regions. In this context, condition \eqref{subeq: unsafe} should be repeated for all unsafe sets (cf. both case studies with multiple unsafe regions).
\end{remark}

To demonstrate the effectiveness of the R-CBC in ensuring the safety of dt-IAPS as defined in Definition~\ref{safey}, we raise the following theorem, which has been adapted from \cite{10066195}.

\begin{theorem}[\textbf{Robust Safety Guarantee}]\label{thm: model-based}
Given a dt-IAPS, let $\mathcal{B}$ be an R-CBC for $\Theta$ as defined in Definition~\ref{def: R-CBC}, with $c $ satisfying \eqref{eq: climit}. Then, for any $x_0 \in X_0$ and $k \in \mathbb{N}$ under input and disturbance signals $u(\cdot)$ and $w(\cdot)$, one has $x_{x_{0uw}}(k) \notin X_1$.
\end{theorem}

\begin{proof}
According to \eqref{subeq: initial}, the initial state satisfies $\mathcal{B}(x(0)) \leq \gamma_1 < \gamma_2$. We now demonstrate that if $\mathcal{B}(x(k)) < \gamma_2$, then $\mathcal{B}(x(k+1)) < \gamma_2$. Given that $c \leq \gamma_2(1 - \lambda)$ as per \eqref{eq: climit}, and in line with \eqref{subeq: decreasing}, we have $\mathcal{B}(x(k+1)) < \lambda \gamma_2 + \gamma_2(1 - \lambda) = \gamma_2$. Therefore, according to \eqref{subeq: unsafe}, it follows that $x(k+1) \notin X_1$. Consequently, all system trajectories will remain outside $X_1$ over an infinite time horizon, thus completing the proof.
\end{proof} 

While the R-CBC defined in Definition~\ref{def: R-CBC} can ensure robust safety guarantees as stated in Theorem~\ref{thm: model-based}, designing an R-CBC is infeasible due to the unknown models that appear on the left-hand side of \eqref{subeq: decreasing} (\emph{i.e.,} $x^+ = A\mathcal{R}(x) + B\mathcal{G}(x)u + w$). With this critical challenge in mind, we now formally present the problem we aim to address in this work.

\color{black}
\begin{resp}
\begin{problem}
Consider a dt-IAPS $\Theta$, as defined in Definition~\ref{def: dt-IAPS}, with an unknown matrices $A, B$, unknown  $\mathcal{R}(x), \mathcal{G}(x)$ with an upper bound on their maximum degree, unknown-but-bounded disturbance $w$, and a safety specification $\Upsilon = (X_0, X_1)$. Collect input-state data from dt-IAPS and develop a direct data-driven framework for designing an R-CBC and its corresponding R-SC, as outlined in Definition~\ref{def: R-CBC}, thereby guaranteeing $\Theta \models \Upsilon$, as demonstrated in Theorem~\ref{thm: model-based}.
\end{problem}
\end{resp}

\section{Data-Driven Construction of R-CBCs and R-SCs}\label{sec: data scheme}
In this section, we propose our data-driven method for constructing an R-CBC along with its corresponding R-SC for the unknown dt-IAPS, as in Definition~\ref{def: dt-IAPS}. To achieve this, we consider the R-CBC structure in a quadratic form, $\mathcal{B}(x) = x^\top P x$, where $P \succ 0$. We gather input-state data (\emph{i.e.,} measurements) over a time horizon of $\left[0, T\right]$, where $T \in \mathbb{N}^+$ denotes the total number of collected samples:

\begin{subequations}\label{eq: data}
\begin{align}
\mathbb{U}_{0,T} &:= [u(0) \;\; u(1) \;\; \dots \;\; u(T-1)], \label{eq:U_0_T} \\
\mathbb{X}_{0,T} &:= [x(0) \;\; x(1) \;\; \dots \;\; x(T-1)], \label{eq:X_0_T} \\
\mathbb{W}_{0,T} &:= [w(0) \;\; \!\!w(1)\;\; \dots \;\; w(T-1)], \label{eq:E_0_T} \\
\mathbb{X}_{1,T} &:= [x(1) \;\;x(2) \;\; \dots \;\; x(T)]. \label{eq:X_1_T}
\end{align}
\end{subequations}
The trajectory $\mathbb{X}_{1,T}$ is precisely collected in discrete time as the system's unknown dynamics evolve recursively at each interval. Additionally, $\mathbb{W}_{0,T}$ is fully unknown. Given the collected trajectories, we now compute the following auxiliary trajectories that are required for our analysis:
\begin{subequations}\label{eq:V_0_T11}
\begin{align}\notag
\mathbb{G}_{0,T} &:= [\mathcal G(x(0))u(0) \;\; \mathcal G(x(1))u(1) \;\; \dots\;\;\\\label{eq:V_0_T}&~~~~~~~~~~~~~~~~~~~~~~~~~~~~~~~~~~  \mathcal G(x(T-1))u(T-1)],\\\label{eq:N_0_T}
\mathbb{R}_{0,T} &:= [\mathcal R(x(0)) \;\;  \mathcal R(x(1)) \;\;  \dots \;\;  \mathcal R(x(T-1))].
\end{align}
\end{subequations}
Drawing on the insights from~\cite{de2019formulas,guo2021data}, we establish the required conditions that form the basis for the data-driven representation of dt-IAPS, as presented in the next lemma.

\begin{lemma}\label{lemma1}
Let \( \mathbb{Q}(x) \in \mathbb{R}^{T \times n} \) be a matrix polynomial satisfying
\begin{equation}\label{con1}
\mathbb{L}(x) = \mathbb{R}_{0,T}  \mathbb{Q}(x), 
\end{equation}
where $\mathbb{R}_{0,T} $ is a full-row rank matrix as in \eqref{eq:N_0_T} and $\mathbb{L}(x)\in \mathbb{R}^{N \times n}$ is a transformation matrix fulfilling
\begin{equation}\label{transform}
\mathcal{R}(x)=\mathbb{L}(x)x.  
\end{equation}
By designing the state-feedback controller  
\begin{align}\label{cont}
	u = \mathcal{K}(x)x, \quad \text{with} \quad \mathcal{K}(x) = \mathbb{U}_{0,T} \mathbb{Q}(x),
\end{align}
the system~\eqref{sys} can be equivalently represented using collected data as
\begin{equation}\label{lem}
x^+= \CLbr x + w,
\end{equation}
where $\tilde{A} = [A\quad B], \quad \tilde{\mathbb{R}}(x) = \left[\begin{array}{c} \mathbb{R}_{0,T} \\
\mathcal{G}(x)\mathbb{U}_{0,T}\end{array}\right]\!\!.$
\end{lemma}
\begin{proof}
Using the collected data in \eqref{eq: data}, it is clear that  
\begin{align}\label{new34}
	 \mathbb{X}_{1,T} =  A\mathbb{R}_{0,T} + B\mathbb{G}_{0,T} + \mathbb{W}_{0,T}.
\end{align}
By designing the state feedback \( u= \mathcal{K}(x)x\), one has
\begin{align}\notag
A \mathcal{R}(x) \!+\! B\mathcal{G}(x)u &\overset{\eqref{transform},\eqref{cont}}{=} \!\!(A\mathbb{L}(x)  \!+\! B\mathcal{G}(x)\mathcal{K}(x))x \\\notag&\overset{\eqref{con1},\eqref{cont}}{=} 
\!\!(A \mathbb{R}_{0,T} \mathbb{Q}(x) \!+\! B\mathcal{G}(x) \mathbb{U}_{0,T} \mathbb{Q}(x))x\\\notag
&~~\!= (A \mathbb{R}_{0,T} \!+\! B\mathcal{G}(x) \mathbb{U}_{0,T}) \mathbb{Q}(x)x \\\label{new33}&~~\!= \underbrace{[A\quad B]}_{\tilde{A}}\underbrace{\left[\begin{array}{c} \mathbb{R}_{0,T} \\
\mathcal{G}(x)\mathbb{U}_{0,T}\end{array}\right]}_{\tilde{\mathbb{R}}(x)} \mathbb{Q}(x)x.
\end{align}
Then, we have 
\begin{align}\notag
	x^+ = A \mathcal{R}(x) \!+\! B\mathcal{G}(x)u + w &\overset{\eqref{new33}}{=}  \CLbr x + w,
\end{align}
which concludes the proof.  
\end{proof}
\begin{remark} 
Condition~\eqref{transform} ensures that all expressions are ultimately expressed in terms of $x$ rather than $\mathcal{R}(x)$. This consistency is essential since our R-CBC is defined as $\mathcal{B}(x) = x^\top P x$, which depends solely on $x$. Expressing everything in terms of $x$ simplifies the formulation and facilitates the proof of our main results in Theorem~\ref{thm: main}. Without loss of generality, a dictionary of monomials $\mathcal{R}(x)$ (with $\mathcal{R}(0) = 0$), can always be reformulated based on $\mathbb{L}(x)$ in \eqref{transform}.
\end{remark}
 
\begin{remark}\label{full-row-rank}
A required condition for the existence of $\mathbb{Q}(x)$ in \eqref{con1} is that $\mathbb{R}_{0,T}$ must have full row rank. To achieve this, the number of samples $T$ must exceed $N$ (i.e., $T > N$). This condition is simple to verify as $\mathbb{R}_{0,T}$ is constructed directly from sampled data.
\end{remark}
Building on the data-based results of Lemma~\ref{lemma1}, we propose the following theorem, as the central contribution of this work, enabling the design of an R-CBC and its R-SC using the collected trajectories in~\eqref{eq: data}.

\begin{theorem}[\textbf{Data-driven R-CBC and R-SC}]\label{thm: main}
Consider an unknown dt-IAPS $\Theta$ with its data-based representation, as in Lemma~\ref{lemma1}. Suppose there exists a state-dependent matrix $\mathcal{H}(x) \in \mathbb{R}^{T \times n}$ such that
\begin{equation}\label{H&P}
\mathbb{R}_{0,T} \mathcal{H}(x)= \mathbb{L}(x) P^{-1}.  
\end{equation}
If there exist constants $\gamma_1,\gamma_2\in \mathbb{R}^{+}$, and $\lambda \in (0,1]$ with $\gamma_2 > \gamma_1$ and $(1-\lambda)\gamma_2 \geq c$, where $c = (1 + \frac{1}{\pi})\Vert\sqrt{P}\Vert^2\delta$, such that the following conditions are satisfied:
\begin{subequations}\label{eq: key-conditions}
    \begin{align}
        &x^\top \left[\mathbb{L}(x)^\dagger \mathbb{R}_{0,T} \mathcal{H}(x)\right]^{-1} x \leq \gamma_1, \quad \forall x \in X_0, \label{subeq: thm-ini}\\
        &x^\top \left[\mathbb{L}(x)^\dagger \mathbb{R}_{0,T} \mathcal{H}(x)\right]^{-1} x \geq \gamma_2, \quad \forall x \in X_1, \label{subeq: thm-uns}\\
        &\begin{bmatrix}
            -\lambda P^{-1} & 0 & 0\\
            \star & 0 & \tilde{\mathbb{R}}(x)\mathcal{H}(x)\\
            \star & \star & -(1+\pi)^{-1} P^{-1}
        \end{bmatrix} \notag \\
        &- \alpha(x) \begin{bmatrix}
           \mathbb{X}_{1,T}\mathbb{X}_{1,T}^\top - T\delta\mathds{I}_n & -\mathbb{X}_{1,T}\hat{\mathbb{R}}^\top & 0\\
            \star & \hat{\mathbb{R}}\hat{\mathbb{R}}^\top & 0\\
            \star & \star & 0
        \end{bmatrix} \!\leq\! 0, ~ \forall x \!\in\! \tilde{X}, \label{subeq: thm-dec} \\
        & \text{where }\quad \tilde{\mathbb{R}}(x) = \left[\begin{array}{c} \mathbb{R}_{0,T} \\
\mathcal{G}(x)\mathbb{U}_{0,T}\end{array}\right]\!\!,\quad \hat{\mathbb{R}} = \left[\begin{array}{c} \mathbb{R}_{0,T} \\
\mathbb{G}_{0,T}\end{array}\right]\!\!, \notag
    \end{align}
\end{subequations}
for some $\pi, \alpha(x) \in\mathbb{R}^{+}$, then it follows that $\mathcal{B}(x) = x^\top \overbrace{\left[\mathbb{L}(x)^\dagger \mathbb{R}_{0,T} \mathcal{H}(x)\right]^{-1}}^{P} x$ serves as an R-CBC  and $
u = \mathbb{U}_{0,T} \mathcal{H}(x) \underbrace{\left[\mathbb{L}(x)^\dagger \mathbb{R}_{0,T} \mathcal{H}(x)\right]^{-1}}_{P} x$ is its associated R-SC for the unknown $\Theta$ with $\rho = (1 + \frac{1}{\pi})\Vert\sqrt{P}\Vert^2$, equivalently $c = (1 + \frac{1}{\pi})\Vert\sqrt{P}\Vert^2\delta$.
\end{theorem}
\begin{proof}
We start by showing the fulfillment of conditions~\eqref{subeq: initial}-\eqref{subeq: unsafe}. Based on condition~\eqref{H&P}, it can be directly concluded that satisfying conditions~\eqref{subeq: thm-ini}-\eqref{subeq: thm-uns} ensures compliance with conditions~\eqref{subeq: initial}-\eqref{subeq: unsafe}, where $P = \left[\mathbb{L}(x)^\dagger \mathbb{R}_{0,T} \mathcal{H}(x)\right]^{-1}$.
Next, we demonstrate that condition~\eqref{subeq: decreasing} is also satisfied. 
Utilizing the quadratic form of the R-CBC, we have
\begin{align}\notag
 \mathcal{B}(x^+)\!&=\! \big(A\mathcal{R}(x) \!+\! B\mathcal{G}(x)u \!+\! w\big)^\top \!\!P\big(A\mathcal{R}(x) \!+\! B\mathcal{G}(x)u \!+\! w\big)\\\notag
&\overset{\eqref{lem}}{=}\big(\CLbr x + w\big)^\top P \big(\CLbr x + w\big) \\\notag
 &\,\,=x^\top\big( \CLbr \big)^{\top} P  \big(\CLbr \big)x \\\notag
 &~~~ + 2\underbrace{x^\top\big(\CLbr\big)^{\top}\!\! \sqrt{P}}_a\underbrace{\sqrt{P}w}_b + w^\top P w .
 \end{align}
According to the Cauchy-Schwarz inequality \cite{bhatia1995cauchy}, \emph{i.e.,}  $a b \leq \Vert a \Vert \Vert b \Vert,$ for any $a^\top, b \in \R^{n}$, followed by
employing Young's inequality \cite{young1912classes}, \emph{i.e.,} $\Vert a \Vert \Vert b \Vert \leq \frac{\pi}{2} \Vert a \Vert^2 + \frac{1}{2\pi} \Vert b \Vert^2$, for any $\pi \in \mathbb R^+$, one has
\begin{align}\notag
	\mathcal{B}(x^+)&\leq x^\top\big(\CLbr \big)^{\top} P  \big(\CLbr\big)x \\\notag
	&~~~ + \pi x^\top\big(\CLbr \big)^{\top} P  \big(\CLbr\big)x \\\notag
	& ~~~ + \frac{1}{\pi}\Vert\sqrt{P}\Vert^2\Vert w \Vert^2 +  \Vert\sqrt{P}\Vert^2\Vert w \Vert^2\\\notag
	&=(1 \!+\! \pi) x^\top\big(\CLbr \big)^{\top} P  \big(\CLbr\big)x \\\label{new87}& ~~~ + \underbrace{(1 + \frac{1}{\pi})\Vert\sqrt{P}\Vert^2}_\rho\Vert w \Vert^2. 
\end{align}
Therefore, by subtracting $\lambda\mathcal{B}(x)$ from both sides of \eqref{new87}, and according to \eqref{assum}, one has
  \begin{align*}
 &\mathcal{B}(x^+) - \lambda\mathcal{B}(x) \notag \\ 
 &~\leq (1 + \pi)x^\top\big(\CLbr \big)^{\top} P \big(\CLbr\big)x \\&~~~  -\lambda \underbrace{x^\top P x}_{\mathcal{B}(x)} +\underbrace{ \rho\delta}_c.
\end{align*}
By defining 
 \begin{align}\label{New99}
V(\tilde{A}; x) := (1 + \pi) (\CLbr)^\top P (\CLbr) - \lambda P,
\end{align}
it is clear that 
\begin{equation*}
    \text{if} \,\,\, V(\tilde{A}; x)\leq0 \Rightarrow  \mathcal{B}(x^+) - \lambda\mathcal{B}(x)\leq c.
\end{equation*}
Therefore, it is needed to enforce $V(\tilde{A}; x)\leq0$. By defining
\begin{align}\label{New991}
\mathcal{J}(\tilde{A}; x) \! := \! (1 \! +\! \pi) (\CLbr) P^{-1} (\CLbr)^\top \!\!-\!\! \lambda P^{-1}\!,
\end{align}
it can be shown using Schur complement \cite{caverly2019lmi} that
\begin{equation*}
    V(\tilde{A};x)\leq0 ~\text{in}~ \eqref{New99}\Leftrightarrow \mathcal{J}(\tilde{A};x)\leq 0 ~\text{in}~ \eqref{New991}.
\end{equation*}
\noindent Since, according to~\eqref{H&P},  one has $\mathbb{R}_{0,T} \mathcal{H}(x) P=\mathbb{L}(x)$, we utilize~\eqref{con1} and set $\mathbb{Q}(x)=\mathcal{H}(x) P$, implying that $\mathbb{Q}(x) P^{-1}=\mathcal{H}(x)$.
We now rewrite \eqref{New991} based on $\mathcal H(x)$ as
\begin{align*}
\mathcal{J}(\tilde{A}; x) \! = \!  \tilde{A}( \tilde{\mathbb{R}}(x)\mathcal{H}(x)((1 \! +\!  \pi)P)\mathcal{H}(x)^\top\tilde{\mathbb{R}}(x)^\top )\tilde{A}^\top \!\!-\!\! \lambda P^{-1}\!.
\end{align*}
Then from now onwards, we aim to enforce $\mathcal{J}(\tilde{A};x)\leq 0$. To do so, we rewrite \eqref{new34} as
\begin{align}\label{new35}
	 \mathbb{X}_{1,T} =  \tilde{A}\hat{\mathbb{R}} + \mathbb{W}_{0,T}, \quad \text{with}~\hat{\mathbb{R}} = \left[\begin{array}{c} \mathbb{R}_{0,T} \\
\mathbb{G}_{0,T}\end{array}\right]\!\!.
\end{align}
On the other hand, by applying Schur complement \cite{caverly2019lmi} to  \eqref{assum}, one can verify that
\begin{align*}
	w w^\top \leq \delta \mathds{I}_n, \quad \forall w \in W(\delta).
\end{align*}
Since
\begin{align*}
    &\mathbb{W}_{0,T} \mathbb{W}_{0,T}^\top = \sum_{k=0}^{T-1} w(k) w(k),
\end{align*}
we have
\begin{align}\label{eq: assum_t}
    T\delta \mathds{I}_n \geq &~\mathbb{W}_{0,T} \mathbb{W}_{0,T}^\top \notag \\
    \overset{\eqref{new35}}{=}  & (\mathbb{X}_{1,T}  - \tilde{A}\hat{\mathbb{R}} )(\mathbb{X}_{1,T}  - \tilde{A}\hat{\mathbb{R}} )^\top \notag \\
        =  & \mathbb{X}_{1,T} \mathbb{X}_{1,T} ^\top  - \tilde{A}\hat{\mathbb{R}}\mathbb{X}_{1,T} ^\top - \mathbb{X}_{1,T} \hat{\mathbb{R}}^\top \tilde{A}^\top + \tilde{A}\hat{\mathbb{R}}\hat{\mathbb{R}}^\top \tilde{A}^\top.  \notag \\
\end{align}
By defining
\begin{align*}
\hat{\mathcal{J}}(\tilde{A}) :=& \tilde{A}\hat{\mathbb{R}}\hat{\mathbb{R}}^\top \tilde{A}^\top -  \tilde{A}\hat{\mathbb{R}}\mathbb{X}_{1,T} ^\top - \mathbb{X}_{1,T} \hat{\mathbb{R}}^\top \tilde{A}^\top \\
&+ \mathbb{X}_{1,T} \mathbb{X}_{1,T} ^\top  - T\delta \mathds{I}_n,
\end{align*}
inequality \eqref{eq: assum_t} can be rewritten as $\hat{\mathcal{J}}(\tilde{A})\leq 0$. 
Therefore, according to S-procedure \cite{caverly2019lmi}, in order to satisfy $\mathcal{J}(\tilde{A}; x) \leq 0 $ conditioned on the data conformity in \eqref{eq: assum_t}, it is sufficient that $\exists \,\, \alpha(x) \in \mathbb{R^+} $ such that
\begin{equation*} 
   \mathcal{J}(\tilde{A}; x) - \alpha(x) \hat{\mathcal J}(\tilde{A})  \leq 0,
\end{equation*}
which could be rewritten in matrix form as
    \begin{align} \label{eq: SP-result}
&\left[\begin{array}{c} \mathds{I}_n\\\tilde{A}^\top \end{array}\right]^\top\Big(\!\begin{bmatrix}
            -\lambda P^{-1} & 0 \\
            \star &  (1+\pi)\tilde{\mathbb{R}}(x)\mathcal{H}(x)P\mathcal{H}(x)^\top\tilde{\mathbb{R}}(x)^\top\\
        \end{bmatrix} \notag \\
        &- \alpha(x) \begin{bmatrix}
           \mathbb{X}_{1,T}\mathbb{X}_{1,T}^\top - T\delta\mathds{I}_n & -\mathbb{X}_{1,T}\hat{\mathbb{R}}^\top\\
            \star & \hat{\mathbb{R}}\hat{\mathbb{R}}^\top
        \end{bmatrix} \!\Big)\left[\begin{array}{c} \mathds{I}_n \\\tilde{A}^\top \end{array}\right] \leq 0. 
    \end{align}
The only remaining challenge in satisfying \eqref{eq: SP-result} lies in the bilinear term $\mathcal{H}(x) P \mathcal{H}(x)^\top$. To resolve this, one can use dilation by applying Schur complement \cite{caverly2019lmi} and demonstrate that inequality (\ref{eq: SP-result}) is equivalent to \eqref{subeq: thm-dec}, Therefore, \eqref{subeq: thm-dec} is sufficient to satisfy \eqref{eq: SP-result}, which concludes the proof.
\end{proof}

\begin{remark}
	One can consider $P^{-1} = Z$ in ~\eqref{subeq: thm-dec} and solve for $Z$ such that it is symmetric and positive definite. Once $Z$ is determined, its inverse will provide the matrix $P$. Additionally, since $\lambda$ and $\pi$ in \eqref{subeq: thm-dec} are both scalar, they are fixed a-priori to avoid bilinearity when solving this condition for the design of $P$ and $\mathcal H$. 
\end{remark}

After establishing the main theorem of this work, we now propose the following lemma that transforms conditions~\eqref{subeq: thm-ini}-\eqref{subeq: thm-dec} into a sum-of-squares (SOS) optimization problem, enabling the systematic computation of R-CBCs and their associated R-SCs.

\begin{lemma}[\textbf{SOS Computation of R-CBC and R-SC}]\label{SOS}
Let $X$, $X_0$, and $X_1$ be defined by polynomial inequality vectors as $X = \{x \in \mathbb{R}^n \mid \beta(x) \geq 0\}$, $X_0 = \{x \in \mathbb{R}^n \mid \beta_0(x) \geq 0\}$, and $X_1 = \{x \in \mathbb{R}^n \mid \beta_1(x) \geq 0\}$. Then $\mathcal{B}(x) = x^\top \left[\mathbb{L}(x)^\dagger \mathbb{R}_{0,T} \mathcal{H}(x)\right]^{-1} x$ is an R-CBC for the unknown dt-IAPS~\eqref{sys} with $\rho = (1 + \frac{1}{\pi})\Vert\sqrt{P}\Vert^2$, equivalently $c = (1 + \frac{1}{\pi})\Vert\sqrt{P}\Vert^2\delta$, and $
 u= \mathbb{U}_{0,T} \mathcal{H}(x) \left[\mathbb{L}(x)^\dagger \mathbb{R}_{0,T} \mathcal{H}(x)\right]^{-1} x   
$ is its R-SC if there exist a state-dependent polynomial matrix $\mathcal{H}(x) \in \mathbb{R}^{{T} \times n}$, constants $ \gamma_1, \gamma_2 \in \mathbb{R}^+$, $\lambda \in (0,1]$ with $c \leq  \gamma_2 (1-\lambda)$, and SOS polynomial vectors $\varpi(x), \varpi_0(x),\varpi_1(x)$ such that the following conditions are SOS polynomials:
\begin{subequations}\label{eq: SOS conditions}
        \begin{align}
            &-x^\top \left[\mathbb{L}(x)^\dagger \mathbb{R}_{0,T} \mathcal{H}(x)\right]^{-1} x - \varpi_0^\top(x) \beta_0(x) + \gamma_1,  \label{subeq: lem-ini}\\
            & ~~~~x^\top \left[\mathbb{L}(x)^\dagger \mathbb{R}_{0,T} \mathcal{H}(x)\right]^{-1} x - \varpi_1^\top(x) \beta_1(x) - \gamma_2, \label{subeq: lem-uns}\\\notag
            &-\begin{bmatrix}
            	-\lambda P^{-1} & 0 & 0\\
            	\star & 0 & \tilde{\mathbb{R}}(x)\mathcal{H}(x)\\
            	\star & \star & -(1+\pi)^{-1} P^{-1}
            \end{bmatrix} \\\notag
            &~~+\alpha(x) \begin{bmatrix}
            	\mathbb{X}_{1,T}\mathbb{X}_{1,T}^\top - T\delta\mathds{I}_n & -\mathbb{X}_{1,T}\hat{\mathbb{R}}^\top & 0\\
            	\star & \hat{\mathbb{R}}\hat{\mathbb{R}}^\top & 0\\
            	\star & \star & 0
            \end{bmatrix} \\\label{subeq: lem-dec}
            &~~- \varpi^\top(x) \beta(x)\mathds{I}_{2n+N+\hat{N}}, 
\end{align}
    \end{subequations}
for some $\pi, \alpha(x) \in \mathbb{R}^+$, while condition~\eqref{H&P} is also fulfilled.
\end{lemma}
\begin{proof}
Since $\varpi_0(x)$ is an SOS polynomial, it follows that $\varpi_0^\top(x) \beta_{0}(x) \geq 0$ for all $x \in X_0 = \{x \in \mathbb{R}^{n} \mid \beta_{0}(x) \geq 0\}$. Since $\mathcal{B}(x) = x^\top \left[\mathbb{L}(x)^\dagger \mathbb{R}_{0,T} \mathcal{H}(x)\right]^{-1} x$, where  $P = \left[\mathbb{L}(x)^\dagger \mathbb{R}_{0,T} \mathcal{H}(x)\right]^{-1} \succ 0$, is also a non-negative SOS polynomial, it can be concluded that if condition~\eqref{subeq: lem-ini} is satisfied, condition~\eqref{subeq: thm-ini} will also hold. The same reasoning applies to conditions~\eqref{subeq: lem-uns} and~\eqref{subeq: thm-uns}.
Next, we show that condition~\eqref{subeq: thm-dec} is also satisfied. Since $\varpi(x)$ is an SOS polynomial, it implies that $\varpi^\top(x) \beta(x) \geq 0$ for all $x \in X = \{x \in \mathbb{R}^{n} \mid \beta(x) \geq 0\}$ . Given that condition~\eqref{subeq: lem-dec} is also an SOS polynomial, one has
\begin{align*}
\text{condition}~\eqref{subeq: lem-dec} \geq 0.
\end{align*}
Thus, since $ \tilde{X} \subset X$, satisfying~\eqref{subeq: lem-dec} ensures that condition~\eqref{subeq: thm-dec} , completing the proof.  
\end{proof}

We introduce Algorithm~\ref{alg}, which details the required process for designing an R-CBC and its R-SC.

\begin{remark}\label{asset}
To effectively enforce conditions~\eqref{subeq: lem-ini}-\eqref{subeq: lem-dec}, one can leverage established software tools from the literature, such as \textsf{SOSTOOLS} \cite{prajna2004sostools}, in conjunction with semi-definite programming (SDP) solvers like \textsf{SeDuMi} \cite{sedumi}.
\end{remark}

\begin{remark}\label{Choice}
It is worth highlighting that the exact choice of \( \mathcal{R}(x) \) is not required in our framework. Instead, we construct a dictionary for $\mathcal{R}(x)$ based on its maximum degree, encompassing all potential terms in the true system's dynamics, even if some of those terms prove to be redundant (cf. monomials in both case studies).
\end{remark}

\begin{algorithm}[t!]
    \caption{Data-driven design of R-CBC and its R-SC}\label{alg}
    \begin{algorithmic}[1]
        \REQUIRE Safety specification $\Upsilon = (X_0, X_1)$, $\delta$ as in \eqref{assum}, a maximum degree for $\mathcal{R}(x)$ (cf. Remark~\ref{Choice}), and a choice of $\mathcal G(x)$
        \STATE Collect $\mathbb{U}_{0,T},\mathbb{X}_{0,T},  \mathbb{W}_{0,T}, \mathbb{X}_{1,T}$ as in~\eqref{eq: data}
        \STATE Form $\mathbb{G}_{0,T}, \mathbb{R}_{0,T}$ as in~\eqref{eq:V_0_T11} (cf. Remarks~\ref{full-row-rank}), and $ \mathbb{L}(x)$~\eqref{transform}
        \STATE Initialize $\lambda \in (0,1]$ and $\pi \in \mathbb{R}^{+}$\label{init-lambda}
        \STATE Solve~\eqref{subeq: lem-dec} and~\eqref{H&P} using \textsf{SeDuMi} and \textsf{SOSTOOLS}~(cf. Remark~\ref{asset}) for {$P$\footnotemark[1] and $\mathcal{H}(x)$}
        \STATE Construct $\mathcal{B}(x)=x^\top P x$ using $P$, and compute $c = (1 + \frac{1}{\pi})\Vert\sqrt{P}\Vert^2\delta$
        \STATE Compute level sets~$\gamma_1$ and $\gamma_2$ according to \eqref{subeq: thm-ini} and \eqref{subeq: thm-uns}
        \ENSURE Guaranteed robust safety over an infinite time horizon with $c \leq  \gamma_2 (1-\lambda)$, R-CBC~$\mathcal{B}(x) = x^\top \left[\mathbb{L}(x)^\dagger \mathbb{R}_{0,T} \mathcal{H}(x)\right]^{-1} x$, R-SC $
        u\!=\! \mathbb{U}_{0,T} \mathcal{H}(x) \! \left[\mathbb{L}(x)^\dagger \mathbb{R}_{0,T} \mathcal{H}(x)\right]^{-1} \! x   
        $
    \end{algorithmic}
\end{algorithm}
\footnotetext[1]{To satisfy condition~\eqref{subeq: lem-dec}, we set ${Z} = P^{-1}$, ensuring ${Z} \succ 0$. Once ${Z}$ is designed, $P$ is obtained as $P={Z}^{-1}$.}
\section{Case Studies}
To assess the efficacy of our data-driven approach, we applied it to two case studies: an academic system~\cite{6509432} and a Lorenz system~\cite{lopez2019synchronization}, both with \emph{unknown dynamics} and subjected to \emph{unknown-but-bounded} disturbances. Details of each case study are presented below. All simulations were performed on a $\mathsf{MacBook}$ equipped with an $\mathsf{M2\,chip}$ and $\mathsf{32\,GB}$ of memory.

\textbf{Academic system.} Consider the following dt-IAPS
\begin{equation}\label{exp:academic}
\Theta\!:\!\begin{cases}
{x}_1^+=x_1 + \tau(-x_1+x_1x_2+x_2 u) + w_1 \\
{x}_2^+=x_2 + \tau(x_1+2 x_2+x_1^2+x_1^2 x_2+u) + w_2,
\end{cases}
\end{equation}
with $\tau = 0.002$ being the sampling time, which is in the form of \eqref{sys} as 
\begin{align*}
A \!&=\! \begin{bmatrix}
  1-\tau & 0 & 1  & 0 & 0  \\
  \tau & 2\tau +1 & 0 & \tau & \tau
\end{bmatrix}\!\!,\:B \!=\! \begin{bmatrix}
  \tau & 0   \\
  0 & \tau
\end{bmatrix}\!\!,\\\mathcal{R}(x)\!&=\!\begin{bmatrix}
  x_1; x_2; x_1 x_2; x_1^2; x_1^2x_2  
\end{bmatrix}\!\!,\: \mathcal{G}(x)\!=\!\begin{bmatrix}
     x_2\\
    1
\end{bmatrix}\!\!.
\end{align*}
We assume that matrices $A, B$,  and \( \mathcal{R}(x) \) and \( \mathcal{G}(x) \) are unknown, with only the maximum degree of $3$ provided for \( \mathcal{R}(x) \)  and \( \mathcal{G}(x) \). Now we construct two dictionaries for \( \mathcal{R}(x) \) and \( \mathcal{G}(x) \) to encompass all possible combinations of monomials up to degree of $3$ as $\mathcal{R}(x)= [x_1;\! x_2; x_1 x_2; x_1^2 x_2; x_1^2; x_2^2;x_2^3;x_1 x_2^2;x_1^3]$ and $\mathcal{G}(x)= [1;x_1;\! x_2; x_1 x_2; x_1^2 x_2; x_1^2; x_2^2;x_2^3;x_1 x_2^2;x_1^3]$. The upper bound for the norm of disturbance $w$ is given as \( \delta = 2 \times 10^{-4} \). The primary aim is to design an R-CBC and its associated R-SC for the system in \eqref{exp:academic}.

The regions of interest are given as $X = [-10, 10]^2$, $X_0 = [-1, 1]^2$, $X_1 = [-5, -3]^2 \cup [-5, -3] \times [3,  5] \cup[2, 5]^2 $. We follow the steps of Algorithm~\ref{alg} by collecting input-state trajectories over the horizon of $T=14$. Then we compute $\mathcal{H}(x)$ 
and $P$ as\footnotemark[2]
\begin{equation*}
P = 10^{4} \times \begin{bmatrix}
7.6150 & 0.4636\\ 0.4636 & 5.8190  \end{bmatrix}\!\!.    
\end{equation*}
\footnotetext[2]{Matrix $\mathcal{H}(x)$ is not reported due to lack of space.}
Given the constructed \( \mathcal{B}(x) \) based on the matrix $P$ and compliance with the conditions in~\eqref{subeq: lem-ini} and~\eqref{subeq: lem-uns},  we design \(\gamma_1 =    1.4361\times 10^{5}\), and \(\gamma_2 =  5.7445 \times 10^{5}\),  while ensuring \( c \leq  \gamma_2 (1-\lambda) \) in~\eqref{eq: climit}, with \( \pi = 10^{-5} \), \( \lambda = 0.99 \), $\rho =    1.4257 \times 10^{5}$, and \( c =   28.5147\). Then 
\begin{equation*}
\mathcal{B}(x)= 76150.4146\,x_{1}^{2} + 9272.7485\,x_{1}x_{2} + 58190.0147\,x_{2}^{2},
\end{equation*}
is an R-CBC for the dt-IAPS with its R-SC as
\begin{align}\notag
u &= - 17.0273\,x_{1}^{2} - 18.2178\,x_{1}x_{2}\\\label{controller-data}&~~~ - 67.1658\,x_{2}^{2} - 64.5998\,x_{1} - 179.8947\,x_{2}.
\end{align}
Then, according to Theorem \ref{thm: model-based}, we guarantee that the unknown dt-IAPS is \emph{robustly} safe against the unknown-but-bounded disturbance over an infinite time horizon. The simulation results for this case study are provided in Figure~\ref{fig:b1}. The computation time for solving our conditions was $7$ seconds.

\begin{figure}[t!]
	\centering
 \includegraphics[width=0.51\linewidth]{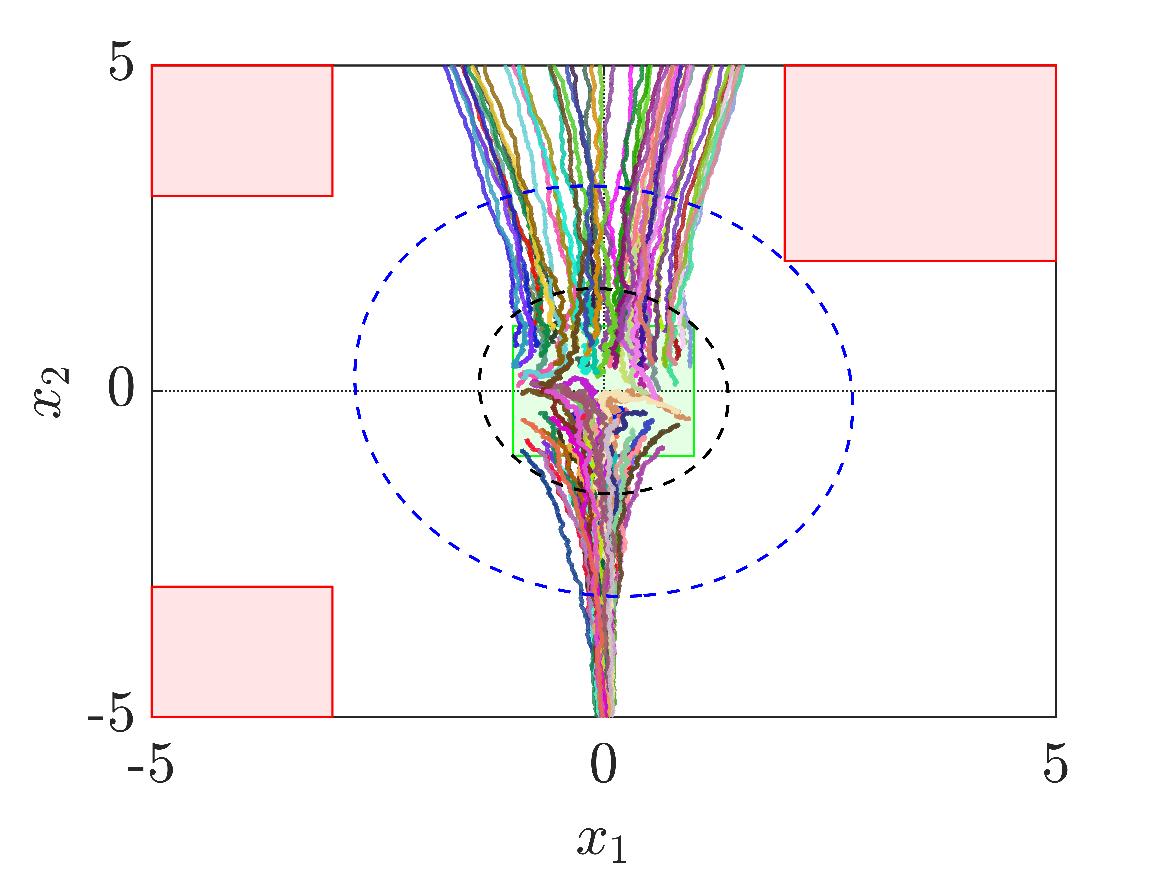}\hspace{-0.3cm}
\includegraphics[width=0.51\linewidth]{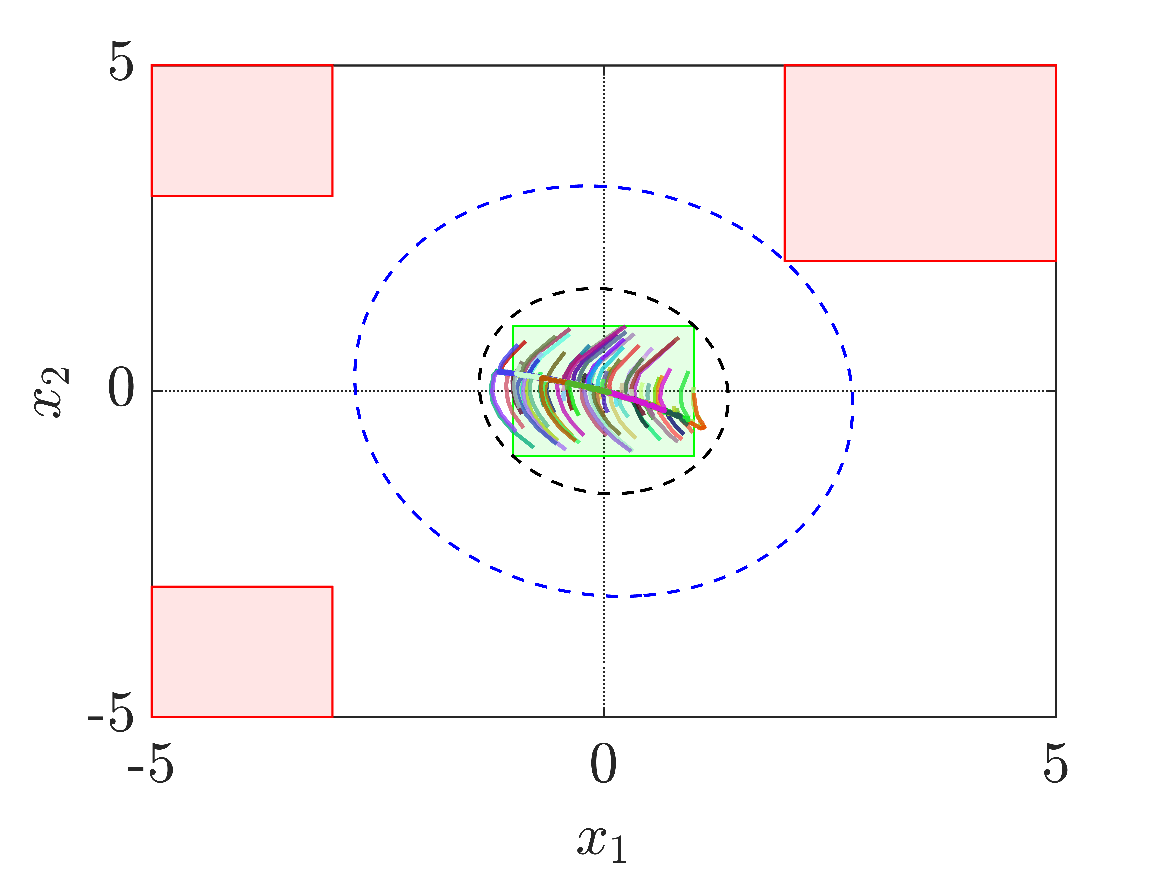}
\includegraphics[width=0.65\linewidth]{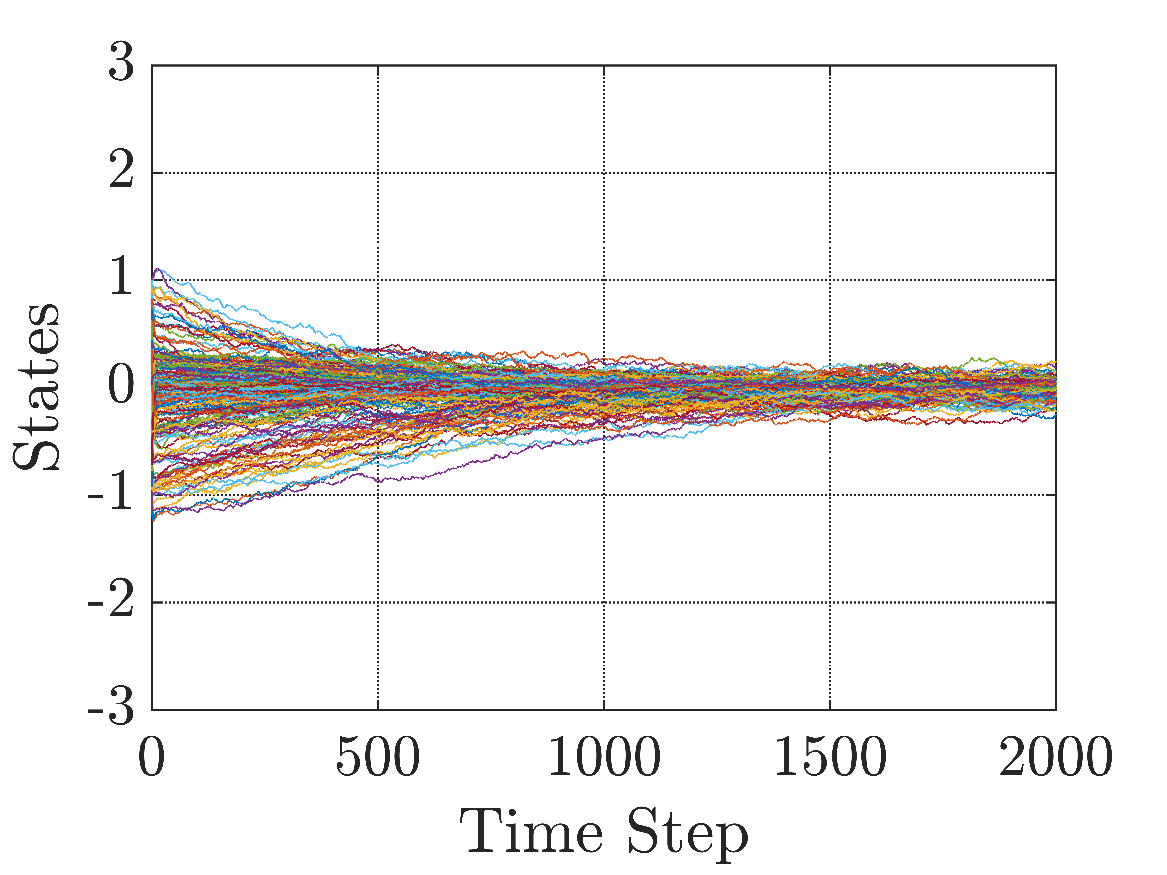}
	\caption{$100$ state trajectories of unknown system~\eqref{exp:academic} without~(\textbf{top-left figure}) and with~(\textbf{top-right figure}) the designed controller~\eqref{controller-data}, starting from different initial conditions in~$X_0 \in  [-1, 1]^2 $. The simulations are generated with $100$ different arbitrary disturbance trajectories satisfying \eqref{assum}, indicating the robustness of our framework to disturbances. Initial~$X_0$ and unsafe~$X_1$ regions are depicted by green \protect\greensquare\ and red \protect\redsquare\ boxes, while $\mathcal B(x) = \gamma_1$ and $\mathcal B(x) = \gamma_2$ are indicated by~\protect\dashedline{0.73cm} and~~\protect\dashedlinea{0.73cm}. The \textbf{bottom figure} indicates trajectories of the unknown system over a time horizon of $2000$, which fulfill the given safety property.}
	\label{fig:b1}
\end{figure}

\textbf{Lorenz system.} We apply our data-driven approach to a 3-dimensional Lorenz system, a classic \emph{chaotic} model widely employed to study complex behaviors in nonlinear systems. The dynamics are described by 
\begin{equation}\label{exp:lorenz}
\Theta\!:\!\begin{cases}
{x}_1^+=x_1+ \tau(10x_2 -10x_1) + w_1 \\
{x}_2^+=x_2 + \tau(28x_1 - x_2- x_1 x_3+u) + w_2,\\
{x}_3^+=x_3+ \tau(x_1x_2 - \frac{8}{3}x_3) + w_3,
\end{cases}
\end{equation}
with $\tau = 0.009$ being the sampling time, which is in the form of \eqref{sys} as 
\begin{align*}
A \!&=\!\! \begin{bmatrix}
  1-10\tau \!&\! 10\tau \!&\! 0  \!&\! 0 \!&\! 0  \\
  28\tau \!&\! 1-\tau \!&\! 0 \!&\! -\tau \!&\! 0 \\
    0 \!&\! 0 \!&\! 1-\frac{8}{3}\tau \!&\! 0 \!&\! \tau
\end{bmatrix}\!\!,~~
B \!=\! \begin{bmatrix}
   0 \\
   \tau\\
    0 
\end{bmatrix}\!\!,
\\\mathcal{R}(x)\!&=\!\begin{bmatrix}
  x_1; x_2; x_3; x_1x_3; x_1x_2  
\end{bmatrix}\!\!,~~\mathcal{G}\!=\!1.
\end{align*}
Here, we assume that matrices $A,B$ and  \( \mathcal{R}(x) \), are unknown, with only the maximum degree of $2$ provided for \( \mathcal{R}(x) \) and $\mathcal{G}\!=\!1$. We structure a dictionary for \( \mathcal{R}(x) \) to encompass all possible combinations of monomials up to degree of $2$ as $\mathcal{R}(x) = [ x_1; x_2; x_3; x_1x_2; x_2x_3; x_1x_3; x_1^2; x_2^2; x_3^2 ]$. The upper bound for the norm of disturbance $w$ is given as \( \delta = 3\times10^{-4} \).

The regions of interest are given as $X = [-10, 10]^3$, $X_0 = [-1.5, 1.5]^3$, and $X_1 = [2.5, 5] \times[-5, -2.5]^2  \cup [-5, -2.5]^3 \cup [2, 5]^2 \times [3.5, 5] $. Following the steps of Algorithm~\ref{alg}, we  gather input-state trajectories~\eqref{eq:U_0_T}-\eqref{eq:X_1_T} over a horizon $T = 15$. Subsequently, we compute $\mathcal{H}(x)$ and $P$, with $P$ reported as
\begin{equation*}
P = 10^{4}\times\begin{bmatrix}
6.0470 & -0.265 & 0.1071\\ -0.2650 & 3.7010 & -0.2609\\ 0.1071 & -0.2609 & 11.53 \end{bmatrix} \!\!. 
\end{equation*}
Given the constructed \( \mathcal{B}(x) \) based on the matrix $P$ and according to conditions in~\eqref{subeq: lem-ini} and~\eqref{subeq: lem-uns},  we design \(\gamma_1 =  5.0712 \times 10^{5}\), and \(\gamma_2 =    1.2772 \times 10^{6}\), while ensuring \( c \leq  \gamma_2 (1-\lambda) \) in~\eqref{eq: climit}, with \( \pi = 10^{-5} \), \( \lambda = 0.99 \), $\rho = 1.6841 \times 10^{5}$, and \( c =   50.5236\). Then
\begin{align*}
\mathcal{B}(x)\!&=\! 60468.2598\,x_{1}^{2} - 5299.1718\,x_{1}x_{2} + 2141.9582\,x_{1}x_{3}\\&~~~ + 37008.6216\,x_{2}^{2} - 5218.5449\,x_{2}x_{3}\\&~~~ + 115251.6783\,x_{3}^{2},
\end{align*}
is an R-CBC for the unknown system with its R-SC as
\begin{align}\notag
u \!&=\!4.0446\,x_{1}^{2} + 11.3915\,x_{1}x_{2} + 21.2063\,x_{1}x_{3}\\\notag &~~~ - 2.5507\,x_{2}^{2} + 53.2232\,x_{2}x_{3} - 2.4537\,x_{3}^{2}\\ \label{controller-data-1} &~~~ - 17.482\,x_{1} - 92.7513\,x_{2} - 13.281\,x_{3}.
\end{align}
According to Theorem \ref{thm: model-based}, we guarantee that the Lorenz system is \emph{robustly} safe against the unknown-but-bounded disturbance over an infinite time horizon. Under the designed robust safety controller, all trajectories of the Lorenz system remain within the safe set $X \backslash X_1$, as can be seen in Figure~\ref{fig:b2}. It is worth noting that our proposed conditions were solved in just $1.8$ seconds, demonstrating their scalability.
 
\begin{figure}[t!]
\centering
 \includegraphics[width=0.50\linewidth]{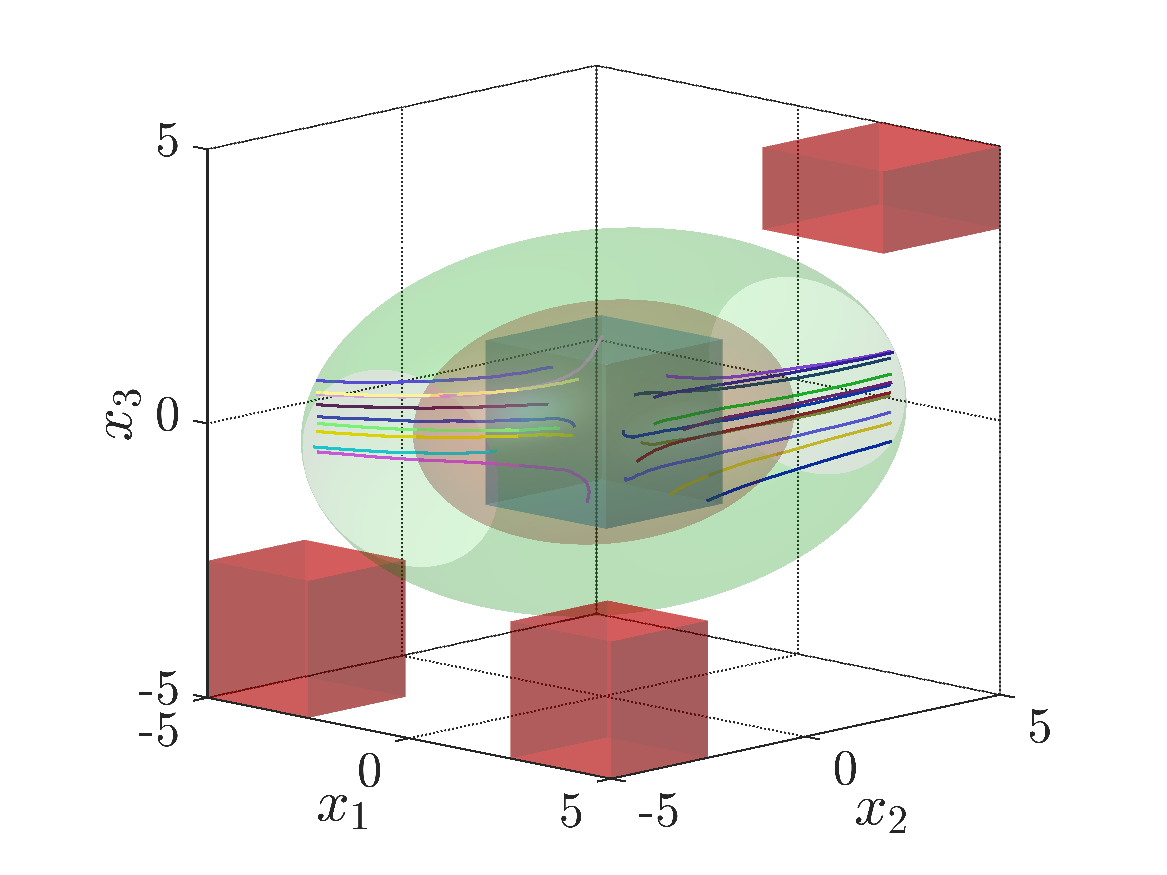}\hspace{-0.3cm}
\includegraphics[width=0.50\linewidth]{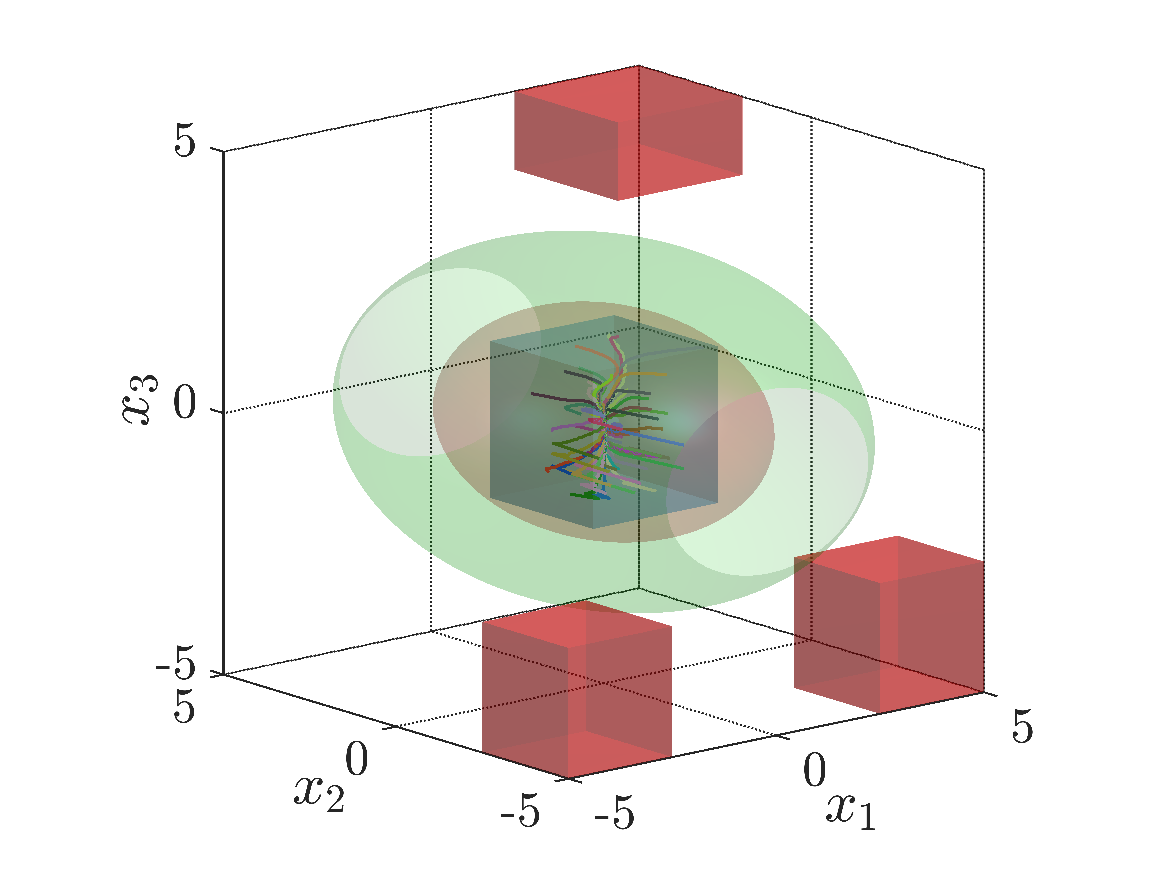}
 \includegraphics[width=0.65\linewidth]{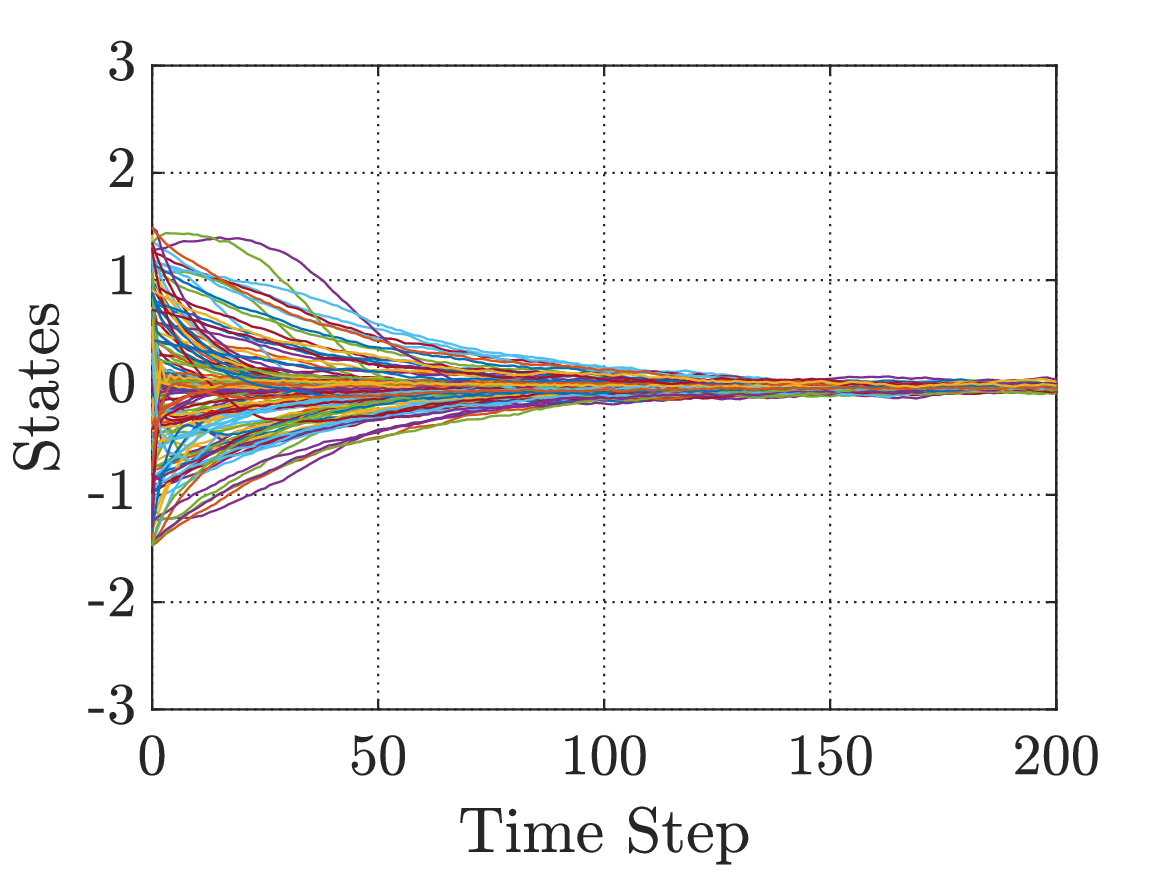}
 	\caption{$50$ state trajectories of unknown Lorenz system~\eqref{exp:lorenz} without~(\textbf{top-left figure}) and with~(\textbf{top-right figure}) designed controller~\eqref{controller-data-1}, starting from different initial conditions in~$X_0 \in  [-1.5, 1.5]^2$. The simulations are performed using $50$ distinct arbitrary disturbance trajectories satisfying \eqref{assum}, demonstrating the robustness of our framework to disturbances. Initial~$X_0$ and unsafe~$X_1$ regions are depicted by blue \protect\bluesquare\ and red \protect\reddsquare\ boxes, while $\mathcal {B}(x) = \gamma_1$ and $\mathcal {B}(x) = \gamma_2$ are indicated by~\protect\greensquare\ and~\protect\redsquare\,, respectively. The \textbf{bottom figure} indicates trajectories of the unknown system over a time horizon of $200$, consistent with the safety property.
 }
\label{fig:b2}
\end{figure}
\section{Conclusion}
We developed a data-driven approach for learning robust control barrier certificates (R-CBCs) and their robust safety controllers (R-SCs) for discrete-time input-affine polynomial systems with unknown dynamics affected by unknown-but-bounded disturbances. By utilizing input-state data over a finite period and ensuring a rank condition for persistent excitation, we directly synthesized R-CBCs and R-SCs from the data, guaranteeing robust system safety without explicit dynamic modeling. We framed the approach as a sum-of-squares (SOS) optimization problem, offering a structured and computationally efficient design framework. Two case studies demonstrated the method's effectiveness in providing robust safety guarantees for unknown input-affine polynomial systems subjected to \emph{bounded disturbances}. Extending our robust safety concept to a broader class of nonlinear systems with unknown dynamics beyond polynomials is being explored as a future direction.

\bibliographystyle{ieeetr}
\bibliography{biblio}

\end{document}